\documentclass{llncs}[10point]
\usepackage{enumitem}
\usepackage[lncs,autoqed]{pamas}
\usepackage{arydshln}
\usepackage{wrapfig}
\usepackage{hyperref}
\usepackage{paralist}
\usepackage{diagbox}
\newcommand{\midd}{\mathbin{:}}
\usepackage{nicefrac}
\usepackage{enumitem,lipsum}

\newcommand\eat[1]{}

\usepackage{tikz}
\usetikzlibrary{arrows}

  \usepackage{algorithm}
  \usepackage{algorithmic}
\usepackage{calrsfs}

	\newcommand{\pref}{\succsim\xspace}
	\newcommand{\Pref}[1][]{
		\ifthenelse{\equal{#1}{}}{\mathrel \succsim}{\mathop{R_{#1}}}
	}    
				\newcommand{\spref}{\ensuremath{\succ}}                                      
	\newcommand{\sPref}[1][]{                  
		\ifthenelse{\equal{#1}{}}{\mathrel \succ}{\mathop{P_{#1}}}
	}                                          
	\newcommand{\Indiff}[1][]{                 
		\ifthenelse{\equal{#1}{}}{\mathrel \sim}{\mathop{\sim_{#1}}}
	}
	\newcommand{\prefset}[1][]{\ifthenelse{\equal{#1}{}}{\mathcal{R}}{\mathcal{R}_{#1}}}
	
	\newcommand{\indiff}{\mathbin \sim\xspace}

	\newcommand{\bor}[1][]{\ifthenelse{\equal{#1}{}}{\mathit{BOR}}{\mathit{BOR}(#1)}}
	
	\newcommand{\egal}{\ensuremath{ESR}\xspace}

	\newcommand{\PS}{MR\xspace}
	
			\newcommand{\ml}[1][]{\ensuremath{\ifthenelse{\equal{#1}{}}{\mathit{ML}}{\mathit{ML}(#1)}}\xspace}
			\newcommand{\sml}[1][]{\ensuremath{\ifthenelse{\equal{#1}{}}{\mathit{SML}}{\mathit{SML}(#1)}}\xspace}
			\newcommand{\sd}[1][]{\ensuremath{\ifthenelse{\equal{#1}{}}{\mathit{SD}}{\mathit{SD}(#1)}}\xspace}
			\newcommand{\rsd}[1][]{\ensuremath{\ifthenelse{\equal{#1}{}}{\mathit{RSD}}{\mathit{RSD}(#1)}}\xspace}
			\newcommand{\rd}[1][]{\ensuremath{\ifthenelse{\equal{#1}{}}{\mathit{RD}}{\mathit{RD}(#1)}}\xspace}
			\newcommand{\st}[1][]{\ensuremath{\ifthenelse{\equal{#1}{}}{\mathit{ST}}{\mathit{ST}(#1)}}\xspace}
			\newcommand{\bd}[1][]{\ensuremath{\ifthenelse{\equal{#1}{}}{\mathit{BD}}{\mathit{BD}(#1)}}\xspace}
			\newcommand{\pc}[1][]{\ensuremath{\ifthenelse{\equal{#1}{}}{\mathit{PC}}{\mathit{PC}(#1)}}\xspace}
			\newcommand{\dl}[1][]{\ensuremath{\ifthenelse{\equal{#1}{}}{\mathit{DL}}{\mathit{DL}(#1)}}\xspace}
			\newcommand{\ul}[1][]{\ensuremath{\ifthenelse{\equal{#1}{}}{\mathit{UL}}{\mathit{UL}(#1)}}\xspace}

\newlength{\wordlength}

\usepackage{boxedminipage}
\usepackage{xspace}

\title{Rank Maximal Equal Contribution: a Probabilistic Social Choice Function}
\author{Haris Aziz and Pang Luo}

	 \institute{%
	NICTA and University of New South Wales, Sydney, NSW 2033, Australia
		\email{haris.aziz@data61.csiro.au}\\
}

\author{Haris Aziz and Pang Luo and Christine Rizkallah}

	 \institute{%
	Data61, CSIRO and UNSW, Sydney, NSW 2033, Australia
		\email{\{haris.aziz,pang.luo\}@data61.csiro.au}\\
	\and
	University of Pennsylvania \\
		Philadelphia, United States\\
		\email{criz@seas.upenn.edu}
		}

\begin{document}

\maketitle

\begin{abstract}
When aggregating preferences of agents via voting, two desirable goals are to incentivize agents to participate  in the voting process and then identify outcomes that are Pareto efficient. We consider participation as formalized by Brandl, Brandt, and Hofbauer (2015) based on the stochastic dominance (SD) relation. 
We formulate a new rule called RMEC (Rank Maximal Equal Contribution) that satisfies the strongest notion of participation and is also ex post efficient. 
The rule is polynomial-time computable and also satisfies many other desirable fairness properties.
The rule suggests a general approach to achieving ex post efficiency and very strong participation. 
\end{abstract}

	\section{Introduction}


Making collective decisions is a fundamental issue in multi-agent systems. 
Two fundamental goals in collective decision making are (1) agents should be incentivized to participate and (2) the outcome should be such that there exists no other outcome that each agent prefers. We consider these goals of \emph{participation}~\citep{BrFi83a,Moul88b} and \emph{efficiency}~\citep{Moul03a} in the context of probabilistic social choice.

			In probabilistic social choice, we study \emph{probabilistic social choice functions (PSCFs)} which take as input agents' preferences over alternatives and return a lottery (probability distribution) over the alternatives.\footnote{PSCFs are also referred to as social decision schemes in the literature.} The lottery can also represent time-sharing arrangements or relative importance of alternatives~\citep{Aziz13b,BMS05a}. For example, agents may vote on the proportion of time different genres of songs are played on a radio channel. This type of preference aggregation is not captured by traditional deterministic voting in which the output is a single discrete alternative which may not be suitable to cater for different tastes.

			When defining notions such as participation, efficiency, and strategyproofness, one needs to reason about preferences over probability distributions (lotteries). In order to define these properties, we consider \emph{stochastic dominance} (\sd). A lottery is preferred over another lottery with respect to \sd, if for all utility functions consistent with the ordinal preferences, the former yields as much utility as the latter. 

			Although efficiency and strategyproofness with respect to \sd have been considered in a series of papers~\citep{Aziz13b,AzSt14a,ABBH12a,ABB13d,BMS05a,Cho12a,Gibb77a,Proc10a}, three notions of participation with respect to \sd  were formalized only recently by \citet{BBH15b}.
		The three notions include very strong (participating is strictly beneficial), strong (participating is at least as helpful as not participating) and standard (not participating is not more beneficial). In contrast to deterministic social choice in which the number of possible outcomes are at most the number of alternatives, probabilistic social choice admits infinitely many  outcomes which makes participation even more meaningful: agents may be able to perturb the outcome of the lottery slightly in their favour by participating in the voting process. In spirit of the radio channel example, voters should ideally be able to increase the fractional time of their favorite music genres by participating in the vote to decide the durations.  

			Two central results presented by \citet{BBH15b} were:
			 \begin{inparaenum}[(1)]
				\item there exists a PSCF (RSD---Random Serial Dictatorship) that satisfies very strong \sd-participation and ex post efficiency (Theorem 4, \citep{BBH15b});
				\item There exists a PSCF (uniform randomization over the Borda winners) that satisfies strong \sd-participation and \sd-efficiency (Theorem 7, \citep{BBH15b}).
			 \end{inparaenum}
In this paper, we propose a new rule  that like RSD satisfies very strong \sd-participation and ex post efficiency and has various desirable properties.

					\begin{table}[h!]
						\centering
					\label{tab:compare}
		\scalebox{1}{			\begin{tabular}{lcccccccc}
					\toprule
					&Serial dictator&\rsd&$\sml$&$BO$&$ESR$&$RMEC$\\
					Properties&&\\
						\midrule
					\sd-efficient&+&--&+&+&+&{--}\\
					ex post efficient&+&+&+&+&+&{+}\\
					\midrule
					Very strong \sd-participation&-&+&--&--&--&{+}\\
					Strong \sd-participation&+&+&--&+&--&{+}\\
					\sd-participation&+&+&+&+&+&{+}\\
					\midrule
					Anonymous &--&+&+&+&+&{+}\\
					Proportional share&--&+&--&--&--&{+}\\
					\midrule
					Strategyproof for dichotomous&+&+&+&--&--&+\\
										and strict preferences&&&&&&\\
					\midrule
					Polynomial-time computable&+&--&+&+&+&+\\
					\bottomrule
					\end{tabular}
					}
					\caption{A comparison of axiomatic properties of different PSCFs: \emph{RSD} (random serial dictatorship), \emph{SML} (strict maximal lotteries), $BO$ (uniform randomization over Borda winners), $ESR$ (egalitarian simultaneous reservation)  and $RMEC$ (Rank Maximal Equal Contribution).
					}
					\label{table:summary:egal}
					\end{table}

		\paragraph{Contributions}

Our central contribution is a new probabilistic voting rule called Rank Maximal Equal Contribution Rule (RMEC). RMEC
satisfies very strong SD-participation and ex post efficiency. 
Moreover RMEC is polynomial-time computable and also satisfies other important axioms such as anonymity, neutrality, fair share, and proportional share. Fair share property requires that each agent gets at least $1/n$ of the maximum possible utility. Proportional share is a stronger version of fair share.
Whereas RMEC is ex post efficient, it is not \sd-efficient.

On the other hand, RMEC has two key advantages over RSD the known rule that satisfies very strong \sd participation. Firstly, RMEC is polynomial-time computable\footnote{Unlike other desirable rules such as maximal lotteries~\citep{ABBH12a,Bran13a} and ESR~\citep{AzSt14a}, 
RMEC is relatively simple and does not require any linear programs to find the outcome lottery.} whereas computing the RSD probability shares is \#P-complete.
The computational tractability of RMEC is a significant advantage over RSD especially when PSCFs are used for time-sharing purposes where computing the time shares is important. 
For RSD, it is even open whether there exists an FPRAS (Fully 
Polynomial-time Approximation Scheme) for computing the outcome shares/probabilities. 
Secondly, RMEC is much more efficient in a welfare sense than RSD. In particular, RMEC dominates RSD in the following sense: for any profile on which RMEC is not \sd-efficient, RSD is not \sd-efficient as well.\footnote{This idea of comparing two mechanisms with respect to a property may be of independent interest. When two mechanisms $f$ and $g$ do not satisfy a property $\phi$ in general, one can still say that that $f$ dominates $g$ with respect to $\phi$ if for any instance on which $f$ does not satisfy $\phi$, $g$ does not satisfy it either. We prove that RMEC dominates RSD wrt \sd-efficiency.} In fact we show that for most preference profiles (for which arbitrary lotteries can be \sd-inefficient), RMEC almost always returns an \sd-efficient outcome. In fact, for 4 or less agents and 4 or less alternatives, all RMEC outcomes are \sd-efficient whereas this is not the case for RSD.

Our formulation of RMEC suggests a general computationally-efficient approach to achieving ex post efficiency and very strong \sd-participation. We identify MEC (Maximal Equal Contribution)---a general class of rules that all satisfy the properties satisfied by RMEC: single-valued, anonymity, neutrality, fair share, proportional share, ex post efficiency, very strong \sd-participation, and a natural monotonicity property. They are also strategyproof under strict and dichotomous preferences. 
%

A relative comparison of different probabilistic voting rules is summarized in Table~\ref{table:summary:egal}.


			 \section{Related Work}

			 One of the first formal works on probabilistic social choice is by \citet{Gibb77a}. 
			 The literature in probabilistic social choice has grown over the years 
			 although it is much less developed in comparison to deterministic social choice~\citep{Bran17a}. The main result of \citet{Gibb77a} was that random dictatorship in which each agent has uniform probability of choosing his most preferred alternative is the unique anonymous, strategyproof and ex post efficient PSCF. Random serial dictatorship (\rsd) is the natural generalization of random dictatorship for weak preferences but the \rsd lottery is \#P-complete to compute~\citep{ABB13b}.  RSD is defined by taking a permutation of the agents uniformly at random and then invoking serial dictatorship: each agent refines the working set of alternatives by picking his most preferred of the alternatives selected by the previous agents).

			 \citet{BoMo01a} initiated the use of stochastic dominance to consider various notions of strategyproofness, efficiency, and fairness conditions in the domain of \emph{random assignments} which is a special type of social choice setting. They proposed the 
			 probabilistic serial mechanism --- a desirable random assignment mechanism. 
			 \citet{Cho12a} extended the approach of \citet{BoMo01a} by considering other lottery extensions such as ones based on lexicographic preferences. 

	Participation has been studied in the context of deterministic voting rules in great detail. 
	\citet{BrFi83a} formalized the paradox of a voter having an incentive to not participate for certain voting rules. \citet{Moul88b} proved that Condorcet consistent voting rules are susceptible to a ``no show.'' 
	We point out that no  deterministic voting rule can satisfy very strong participation. Consider a voting setting with two agents and two alternatives $a$ and $b$. Agent 1 prefers $a$ over $b$ and agent 2 prefers $b$ over $a$. Then whatever the outcome of voting rule, one agent will get a least preferred outcome despite participating. The example further motivates the study of PSCFs with good participation incentives.


			 The tradeoff of efficiency and strategyproofness for PSCFs was formally considered in a series of papers~\citep{Aziz13b, AzSt14a, ABBH12a, ABB13d, BMS05a}.
			 \citet{AzSt14a} presented a generalization --- \emph{Egalitarian Simultaneous Reservation} (\egal) --- of the probabilistic serial mechanism to the domain of social choice.  \citet{Aziz13b} proposed the \emph{maximal recursive (\PS)} PSCF which is similar to the random serial dictatorship but for which the lottery can be computed in polynomial time.  

\citet{BBH15d} study the connection between welfare maximization and participation and show how welfare maximization achieves SD-participation. However the approach does not necessarily achieve very strong SD-participation or even strong SD-participation.
			 
			 In very recent work, \citet{GAX17a} presented an elegant rule called \emph{2-Agree} that satisfies very strong \sd-participation, ex post efficiency, and various other properties. However, the rule is defined for strict preferences.\footnote{Under strict preferences, random dictatorship satisfies all the properties examined in this paper.} 

			 		\section{Preliminaries}

			 			Consider the social choice setting in which there is a set of agents $N=\{1,\ldots, n\}$, a set of alternatives $A=\{a_1,\ldots, a_m\}$ and a preference profile $\pref=(\pref_1,\ldots,\pref_n)$ such that each $\pref_i$ is a complete and transitive relation over $A$. Let $\mathcal{R}$ denote the set of all possible weak orders over $A$ and let $\mathcal{R}^N$ denote all the possible preference profiles for agents in $N$.
						Let $\mathcal{F}(\mathbb{N})$ denote the set of all finite and non-empty subsets of $\mathbb{N}$.
			 			We write~$a \pref_i b$ to denote that agent~$i$ values alternative~$a$ at least as much as alternative~$b$ and use~$\spref_i$ for the strict part of~$\pref_i$, i.e.,~$a \spref_i b$ iff~$a \pref_i b$ but not~$b \pref_i a$. Finally, $\indiff_i$ denotes~$i$'s indifference relation, i.e., $a \indiff_i b$ if and only if both~$a \pref_i b$ and~$b \pref_i a$.
			 			The relation $\pref_i$ results in equivalence classes $E_i^1,E_i^2, \ldots, E_i^{k_i}$ for some $k_i$ such that $a\spref_i a'$ if and only if $a\in E_i^l$ and $a'\in E_i^{l'}$ for some $l<l'$. Often, we will use these equivalence classes to represent the preference relation of an agent as a preference list
			 			$i\midd E_i^1,E_i^2, \ldots, E_i^{k_i}$.
			 		For example, we will denote the preferences $a\indiff_i b\spref_i c$ by the list $i:\ \{a,b\}, \{c\}$. 
			 		For any set of alternatives $A'$, we will refer by $\max_{\pref_i}(A')$ to the set of most preferred alternatives according to preference $\pref_i$.

			 		An agent $i$'s preferences are \emph{dichotomous} if and only if he partitions the alternatives into at most two equivalence classes, i.e., $k_i\leq 2$. An agent $i$'s preferences are \emph{strict} if and only if $\pref_i$ is antisymmetric, i.e.
			 		 all equivalence classes have size 1.

			 			Let $\Delta(A)$ denote the set of all \emph{lotteries} (or \emph{probability distributions}) over $A$.
			 			The support of a lottery $p \in \Delta(A)$, denoted by $\text{supp}(p)$, is the set of all alternatives to which $p$ assigns a positive probability, i.e., $\text{supp}(p) = \{x \in A \mid p(x)>0\}$. We will write $p(a)$ for the probability of alternative $a$ and we will represent a lottery as
			 			$p_1a_1 + \cdots + p_ma_{m}$
			 			where $p_j=p(a_j)$ for $j\in \{1,\ldots, m\}$. For $A'\subseteq A$, we will (slightly abusing notation) denote $\sum_{a\in A'}p(a)$ by $p(A')$.

			 			A \emph{PSCF} is a function $f: \mathcal{R}^n \rightarrow \Delta(A)$. If $f$ yields a set rather than a single lottery, we call $f$ a \emph{correspondence}.
			 		Two minimal fairness conditions for PSCFs are \emph{anonymity} and \emph{neutrality}. Informally, they require that the PSCF should not depend on the names of the agents or alternatives respectively.

			 			In order to reason about the outcomes of PSCFs, we need to determine how agents compare lotteries. A \emph{lottery extension} extends preferences over alternatives to (possibly incomplete) preferences over lotteries.
			 			Given $\pref_i$ over $A$, a \emph{lottery extension} extends $\pref_i$ to preferences over the set of lotteries $\Delta(A)$. We now define \emph{stochastic dominance (SD)} which is the most established lottery extension.
				
						Under \emph{stochastic dominance (SD)}, an agent prefers a lottery that, for each alternative $x \in A$, has a higher probability of selecting an alternative that is at least as good as $x$. Formally, $p \pref_i^{\sd} q$ if and only if 
\[\forall y\in\nolinebreak A \colon \sum_{x \in A: x \pref_i y} p(x) \geq \sum_{x \in A: x \pref_i y} q(x).\]
\sd~\citep{BoMo01a} is particularly important because $p \pref^{\sd} q$ if and only if $p$ yields at least as much expected utility as $q$ for any von-Neumann-Morgenstern utility function consistent with the ordinal preferences \citep{Cho12a}.

					\paragraph{Efficiency}

					A lottery $p$ is \emph{${\sd}$-efficient} if and only if there exists no lottery $q$ such that $q \pref_i^{\sd} p$ for all $i\in N$ and $q \spref_i^{\sd} p$ for some $i\in N$. A PSCF is ${\sd}$-efficient if and only if it always returns an ${\sd}$-efficient lottery. A standard efficiency notion that cannot be phrased in terms of lottery extensions is \emph{ex post efficiency}. A lottery is ex post efficient if and only if it is a lottery over Pareto efficient alternatives. 

%


			 	\paragraph{Participation}

				
				\citet{BBH15b} formalized three notions of participation.

			 	%
			 	%
			 	%
			 	%
			 	%
			 	%
			 	%
			 	%
			 	%

		Formally, a PSCF $f$ satisfies \emph{$\sd$-participation} if there exists no $\pref \in \mathcal{R}^N$ for some $N\in\mathcal{F}(\mathbb{N})$ and $i\in N$ such that $f(\pref_{-i}) \mathrel{\succ_i^{\sd}} f(\pref)$. 
			 
			 A PSCF $f$ satisfies \emph{strong $\sd$-participation} if $f(\pref) \mathrel{\pref_i^{\sd}} f(\pref_{-i})$ for all $N\in\mathcal{F}(\mathbb{N})$, $\pref\in\mathcal{R}^N$, and $i\in N$. 
			 	
				A PSCF $f$ satisfies \emph{very strong $\sd$-participation} if  for all $N\in\mathcal{F}(\mathbb{N})$, $\pref\in\mathcal{R}^N$, and $i \in N$, $f(\pref)\mathrel{\pref_i^{\sd}} f(\pref_{-i})$ and 
			 	\[
			 	  f(\pref)\mathrel{\succ_i^{\sd}} f(\pref_{-i}) \text{ whenever }\exists p\in \Delta(A)\colon p\mathrel{\succ_i^{\sd}} f(\pref_{-i}).
			 	\]

	Very strong SD-participation is a desirable property because it gives an agent strictly more expected utility for \emph{each} utility function consistent with his ordinal preferences. We already pointed out that no deterministic voting rule can satisfy very strong \sd-participation. 

				 		 \paragraph{Strategyproofness}

				 		A PSCF~$f$ is ${\sd}$-\emph{ma\-nip\-u\-la\-ble} if and only if there exists an agent $i \in N$ and preference profiles $\pref$ and $\pref'$ with $\pref_j=\pref_j'$ for all $j \ne i$ such that $f(\pref') \spref_i^{\sd} f(\pref)$.
				 		A PSCF is \emph{weakly} ${\sd}$-\emph{strategyproof} if and only if it is not ${\sd}$-manipulable.
				 		It is ${\sd}$-\emph{strategyproof} if and only if $f(\pref) \pref_i^{\sd} f(\pref')$ for all $\pref$ and $\pref'$ with $\pref_j=\pref_j'$ for all $j \neq i$. Note that \sd-strategyproofness is equivalent to strategyproofness in the Gibbard sense. 

\section{Rank Maximal Equal Contribution}

We present Rank Maximal Equal Contribution (RMEC). The rule is based on the notion of rank maximality that is well-established in other contexts such as assignment~\citep{Mich07a,Feat11a}. 

For any alternative $a$, its \emph{corresponding rank vector} is $r(a)=(r_1(a), \ldots, r_m(a))$ where $r_j(a)$ is the number of agents who have $a$ as his $j$-th most preferred alternative.
For a lottery $p$, its \emph{corresponding rank vector} is $r(p)=(r_1(p), \ldots, r_m(p))$ where $r_j(p)$ is $\sum_{i\in N}\sum_{a\in E_i^j}p(a)$.
We compare rank vectors lexicographically. One rank vector $r=(r_1,\ldots, r_m)$ is \emph{better} than $r'=(r_1',\ldots, r_m')$ if for the smallest $i$ such that $r_i\neq r_i'$, it must hold that $r_i>r_i'$.

The notion of rank vectors leads to a natural PSCF: randomize over alternatives that have the best rank vectors. However such an approach does not even satisfy strong \sd-participation. It can also lead to perverse outcomes in which minority is not represented at all:
	Consider the following preference profile. 
	\begin{align*}
	1:&~a,b&\quad
	2:&~a,b&\quad
	3:&~b,a
	\end{align*}
	For the profile, the rank maximal rule simply selects $a$ with probability 1. This is unfair to agent 3 who is in a minority. Agent $3$ does not get any benefit of participating.

%

Let $F(i,A,\pref)$ be the set of most preferred alternatives of agent $i$ that have best rank vector among all his most preferred alternatives. In the RMEC rule, each agent $i\in N$ contributes $1/n$ probability weight to a subset of his most preferred alternatives. Precisely, he gives probability weight $\nicefrac{1}{n|F(i,A,\pref)|}$ to each alternative in $F(i,A,\pref)$. The resultant lottery $p$ is the RMEC outcome. We formalize the RMEC rule as Algorithm~\ref{algo:ec}.
	We view RMEC outcome lottery $p$ as consisting of $n$ components $p_1,\ldots, p_n$ where 
	\[p_i=\sum_{a\in F(i,A,\pref)}\frac{1}{n|F(i,A,\pref)|}a.\]



	\begin{algorithm}[h!]
		  \caption{The Rank Maximal Equal Contribution rule}
		  \label{algo:ec}
		\begin{algorithmic}
			\REQUIRE $(N,A,\pref)$
			\ENSURE lottery $p$ over $A$.
		\end{algorithmic}
		\begin{algorithmic}[1]
			\STATE Initialize probability $p(a)$ of each alternative $a\in A$ to zero.
		\FOR{$i=1$ to $|N|$}
		\STATE Identify $F(i,A,\pref)$ the subset of alternatives in $\max_{\pref_i}(A)$ with the best rank vector.
 		\FOR{each $a\in F(i,A,\pref)$}
 		\STATE $p(a)\longleftarrow p(a)+\nicefrac{1}{(n|F(i,A,\pref)|)}$
		\ENDFOR
		\COMMENT{we will denote by $p_i$ the probability weight of \\$1/n$ allocated by agent $i$ uniformly to alternatives\\ in $F(i,A,\pref)$}
		\ENDFOR
			\RETURN lottery $p$.
		\end{algorithmic}
	\end{algorithm}

 %
 %
 %
 %

	\begin{example}
		Consider the following preference profile. 
		\begin{align*}
			1:&~~ \{a,b,c,f\}, d, e&
			2:&~~ \{b,d\}, e, \{a,c,f\}\\
			3:&~~ \{a,e,f\}, d, b, c&
			4:&~~ c, d, e, \{a,f\}, b\\
			5:&~~ \{c, d\}, \{e, a, b,f\}
			\end{align*}
		
			The rank vectors of the alternatives are as follows:  $a: (2,1,1,1,0)$; $b: (2,1,1,0,1)$; $c: (3,0,1,1,0)$; $d: (2,3,0,0,0)$; $e: (1,2,2,0,0)$; and $f: (2,1,1,1,0)$.

		
			Each agent selects the most preferred alternatives with the best rank vector to give his $1/5$ probability uniformly to the following alternatives:
$1: c$;
$2: d$;
$3: a,f$;
$4: c$; and
$5: c$.

	%
			So the outcome is 
			$\frac{1}{10}a+\frac{3}{5}c+\frac{1}{5}d +\frac{1}{10}f.$
			

	\end{example}

	%
	%
	%
	%
	%

		\section{Properties of RMEC}


		We observe that RMEC is both anonymous and neutral. The RMEC outcome can be computed in time polynomial in the input size. Since the contribution to an alternative by an agent is $1/yn$ for some $y\in \{1,\ldots, m\}$, the probabilities are rational.
		
		

		%


		\begin{proposition}
			RMEC is anonymous and neutral. 
			The RMEC outcome can be computed in polynomial time $O(m^2n)$ and consists of rational probabilities.
			\end{proposition}

		\bigskip


		%
		%
		%
		%
		%
		%
		%
		%
		%
		%
		%
		%
		%
		%
		%
		%

		Next we note that if preferences are strict, then RMEC is equivalent to random dictatorship. 
As a corollary, RMEC satisfies both \sd-efficiency and very strong \sd-participation under strict preferences. More interestingly, RMEC satisfies very strong \sd-participation even for weak orders. 
%

	
	\begin{proposition}
			RMEC  satisfies very strong \sd-participation.
		\end{proposition}
			\begin{proof}
				Let us consider the RMEC outcome $p$ when $i$ abstains and compare it with the RMEC outcome $q$ when $i$ votes.

		When $i$ abstains, agent $j\in N\setminus \{i\}$ contributes probability weight $1/(n-1)$ uniformly to alternatives in $F(j,A,\pref_{-i})$.
		Now consider the situation when $i$ also votes. We want to identify the alternatives $j$ will contribute to. Our central claim is that 
		\emph{for each $a\in F(j,A,\pref)$ and $b\in \max_{\pref_i}(F(j,A,\pref_{-i}))$, it is the case that $a\succsim_i b$.}
 To prove the claim, assume for contradiction that when $i$ votes, $j$ contributes to some alternative $b$ less preferred by $i$ to $a\in \max_{\pref_i}(F(j,A,\pref_{-i})$. But this is not possible because $b$ had at most the same rank as $a$ when $i$ did not vote but since $a\succ_i b$, $a$ will have strictly more rank than $b$ when $i$ votes. Hence when $i$ votes, agent $j$ sends all his probability weight to either alternatives in  $\max_{\pref_i}(F(j,A,\pref_{-i}))$ or alternatives even more preferred by $i$. Thus we have proved the claim. By proving the claim, we have shown that when $i$ participates, any change in the relative contribution of some agent $j\neq i$ is in favour of agent $i$. 



Take any $b\in A$ and consider $\{a\midd a\pref_i b\}$. Assume $j$ is any agent in $N\setminus \{i\}$. If $j$ contributes anything (at most $1/(n-1)$) to $\{a\midd a\pref_i b\}$ when agent $i$ abstains, then when $i$ votes, $j$ will contribute $1/n$ to $\{a\midd a\pref_i b\}$ because of the central claim proved above. Now, for the two scenarios where $i$ votes or abstains, the contribution difference from $j$ to $\{a\midd a\pref_i b\}$ is at most $1/n(n-1)$, and the total contribution difference from $N\setminus \{i\}$ to $\{a\midd a\pref_i b\}$ is at most $1/n$, which would be compensated by the contribution of $i$ to $\{a\midd a\pref_i b\}$ when $i$ votes.
Therefore for each $b\in A$, $q(\{a\midd a\pref_i b\})\geq p(\{a\midd a\pref_i b\}).$ Thus $q\pref_i^{\sd} p$ so RMEC satisfies strong \sd-participation.



		
		We now show that RMEC satisfies \emph{very} strong \sd-participation.  
Suppose that $p=RMEC(N,A,\pref_{-i})$ is such that $p(\max_{\pref_i}(A))<1$. 
 It is sufficient to show that for $q=RMEC(N,A,\pref)$, $q(\max_{\pref_i}(A))>p(\max_{\pref_i}(A))$. If some other agent $j$'s relative contribution changes in favour of agent $i$, we are already done. So let us assume that each $j\neq i$, $F(j,A,\pref_{-i})=F(j,A,\pref)$. When $i$ votes, the total contribution to $\max_{\pref_i}(A)$ by agents other than $i$  is $p(\max_{\pref_i}(A))\frac{n-1}{n}$. The contribution of agent $i$ to $\max_{\pref_i}(A)$ is $\frac{1}{n}$. Hence 
\begin{align*}
&q(\max_{\pref_i}(A))=\frac{n-1}{n}p(\max_{\pref_i}(A))+\frac{1}{n}(1)\\
=&~\frac{n-1}{n}p(\max_{\pref_i}(A))+ \frac{1}{n}(p(\max_{\pref_i}(A))+1-p(\max_{\pref_i}(A)))\\
=&~p(\max_{\pref_i}(A))+\frac{1}{n}(1-p(\max_{\pref_i}(A)))>p(\max_{\pref_i}(A))
\end{align*}
The last inequality holds because we supposed that $p(\max_{\pref_i}(A))<1$ so that $1-p(\max_{\pref_i}(A))>0$.
Thus RMEC satisfies \emph{very} strong \sd-participation.  
\end{proof}
		
%
%
%

The fact that RMEC satisfies very strong \sd-participation is one the central results of the paper. We note here that very strong \sd-participation can be a tricky property to satisfy. For example the following simple variants of RMEC violate even  strong \sd-participation: (1) each agent contributes to a most  preferred Pareto optimal alternative or (2) each agent contributes uniformly to Pareto optimal alternatives most preferred by her. 

\begin{example}
	Consider the following profile.
	\begin{align*}
	1:&~~ a, e, d, f, b, c&
	2:&~~\{b, c, d, f\}, a, e\\
	3:&~~ e, a, \{b,c,d, f\}&
	4:&~~ e, c, \{f, b\}, a, d\\
	5:&~~ e, \{f,b\}, c, a, d
	\end{align*}

	Suppose we break-ties lexicographically and use rule (1). When 1 does not vote then 2 selects $b$ and the lottery is $\nicefrac{1}{4}b+\nicefrac{3}{4}e$
When 1 votes, 2 selects $c$ because $b$ is now not Pareto optimal (f dominates b). 
The new lottery  $\nicefrac{1}{5}a+\nicefrac{1}{5}c+\nicefrac{3}{5}e$ is not even weakly \sd-better for 1 than $\nicefrac{1}{4}b+\nicefrac{3}{4}e$. The same example works if each agent randomizes uniformly over most preferred Pareto optimal alternatives. 
	\end{example}

Next, we prove that RMEC is also ex post efficient i.e., randomizes over Pareto optimal alternatives. 
\begin{proposition}\label{prop:exposteff}
	RMEC is ex post efficient. 
	\end{proposition}
\begin{proof}
	Each alternative $a$ in the support is an alternative that is the most preferred alternative of an agent $i$ with the best rank vector. Suppose the alternative $a$ is not Pareto optimal. Then there exists another alternative $b$ such that $b\succsim_j a$ for all $j\in N$ and $b\succ_j a$ for some $j\in N$. Note that since $a$ is the most preferred alternative of $i$, it follows that $b\sim_i a$. Since $b$ Pareto dominates $a$, $b$ is a most preferred alternative of $i$ with a better rank vector than $a$. But this contradicts the fact that $a$ is a most preferred alternative of $i$ with the best rank vector. 
	\end{proof}




	Although RMEC is ex post efficient, it unfortunately does not satisfy the stronger efficiency property of \sd-efficiency.
	
			
			\begin{example}
				Consider the following preference profile with dichotomous preferences. 
				\begin{align*}
					1,2,3,4:&~~d &
					5,6:&~\{d,c\} &
					7,8:&~~\{d,b\} \\
					9:&~~\{a,b\}&
					10:&~~\{a,c\}
					\end{align*}
					
					The RMEC outcome is $\frac{8}{10}d+\frac{1}{10}c+\frac{1}{10}b$
but is \sd-dominated by $\frac{9}{10}d+\frac{1}{10}a$. 					\end{example}
		
		In the example above, although each agent chooses those most preferred alternatives that are most beneficial to other agents, the agents do not coordinate to make these mutually beneficial decisions. This results in a lack of \sd-efficiency. Although RMEC is not \sd-efficient just like RSD, it has a distinct advantage over RSD in terms of \sd-efficiency.
		
		\begin{proposition}
			 For any profile, if the RSD outcome is \sd-efficient, then the RMEC outcome is also \sd-efficient.
			  Furthermore, there exist instances for which the RSD outcome is not \sd-efficient but the RMEC is not only \sd-efficient but \sd-dominates the RSD outcome. 
			\end{proposition}
	           \begin{proof}

			Due to the result of \citet{ABB14b} that \sd-efficiency depends on the support, it is sufficient to show that $\text{supp}(\textit{RSD}(N,A,\pref))\supseteq \text{supp}(RMEC(N,A,\pref))$.

Now suppose that $a\in \text{supp}(RMEC(N,A,\pref))$.
 We also know that $a\in F(i,A,\pref)$ for some $i\in N$. We prove that $a\in \text{supp}(\textit{RSD}(N,A,\pref))$ by showing that there exists one permutation $\pi$ under which serial dictatorship gives positive probability to $a$. The first agent in the permutation $\pi$ is $i$.

We build the permutation $\pi$ so that $a$ is an outcome of serial dictatorship with respect to $\pi$. The working set $W$ is initialized to $A$. Agent $i$ refines $W$ to $\max_{\pref_{i}(A)}$. Now suppose for contradiction that each remaining agent strictly prefers some other alternative in $W$ to $a$. In that case, $a$ is not the rank maximal alternative from $\max_{\pref_{i}(A)}$ which is a contradiction to $a\in F(i,A,\pref)$. Thus for some agent $j$ not considered yet, $a$ is a most preferred alternative in $W$. 
We can add such an agent to the permutation and let him refine and update $W$. In $W$, $a$ still remains  rank maximal (with respect to agents who have not been added to the permutation) among alternatives in $W$.  
We can continue identifying a new agent who maximally prefers $a$ in the latest version of $W$ and appending the agent to the permutation $\pi$ until $\pi$ is fully specified. Note that $a$ still remains in the working set which implies that
 $a\in \text{supp}(\textit{RSD}(N,A,\pref))$. This completes the proof that if the RSD outcome is \sd-efficient, then the RMEC outcome is also \sd-efficient.
			 
			 \bigskip
			Next we prove the second statement. Consider the following preference profile. 
	           	\begin{align*}
	           	1&: \{a,c\}, b, d&
	           	2&: \{a,d\}, b, c\\
	           	3&:  \{b, c\}, a,d&
	           	4&:  \{b,d\}, a,c
	           	\end{align*}
	           	The unique RSD lottery is $p=\nicefrac{1}{3}\,a+\nicefrac{1}{3}\, b+\nicefrac{1}{6}\,c+\nicefrac{1}{6}d$,
	           	which is $\sd$-dominated by $\nicefrac{1}{2}\,a+\nicefrac{1}{2}\,b$. This was observed by \citet{ABBH12a}.

	          We now compute the RMEC outcome.
	          	The rank vectors are as follows:
	        $a: (2,2,0,0)$; $b: (2,2,0,0)$; $c: (2,0,2,0)$; and $d: (2,0,2,0)$.
The agents choose alternatives as follows:
	          	\begin{inparaenum}[]
	          		\item $1: a$,  
	          		\item $2: a$, 
	          		\item $3: b$,
	          		\item $4: b$
	          	\end{inparaenum}
	
		RMEC returns the following lottery which is \sd-efficient and \sd-dominates the RSD lottery:  $\nicefrac{1}{2}\,a+\nicefrac{1}{2}\,b.$
This completes the proof.	
	          	 \end{proof}
		 
		 \bigskip

		 Although RMEC is not \sd-efficient in general, we give experimental evidence that it returns \sd-efficient outcomes for most profiles. An exhaustive experiment shows that RMEC is \sd-efficient for every profile with 4 agents and 4 alternatives. This is already in constrast to RSD that can return \sd-inefficient outcome for $n=m=4$.  Further experiments show that RMEC is \sd-efficient for almost all the profiles with a larger size. In the experiment, we generated profiles uniformly at random for specified numbers of agents and alternatives and examined whether the corresponding RMEC lottery is \sd-efficient. The results are shown in Table 2. 
		 
  					\begin{table}[h!]
  						\centering
  					\label{tab:compare}
  			\begin{tabular}{ccccccccc}
				\toprule
  		                       \diagbox{$|A|$} {$|N|$} &4&5&6&7&8&&&\\
  					4&10,000&10,000&10,000&9,999&10,000\\
  					5&9,999&10,000&10,000&9,998&9,999\\
  					6&9,999&10,000&9,996&10,000&9,999\\
  					7&10,000&9,999&9,997&9,998&9,999\\
  					8&9,999&9,996&9,998&9,997&9,996\\
  					\bottomrule
  					\end{tabular}

  			\caption{The number of profiles for which the RMEC outcome is \sd-efficient out of 10,000 uniformly randomly generated profiles for specified numbers of agents and alternatives.}
  					\label{table:summary:random}
					
  					\end{table}
		 
		 Note that in general for any given preference profile with some ties, a significant proportion of lotteries are not \sd-efficient. On the other hand, RMEC almost always returns an \sd-efficient lottery. 

    \bigskip

			We say that a lottery satisfies \emph{fair welfare share} if each agent gets at least $1/n$ of the maximum possible expected utility he can get from any outcome. Fair welfare share was originally defined by \citet{BMS05a} for dichotomous preferences. We observe that since RMEC gives at least $1/n$ probability to each agent's first equivalence class, it follows that each RMEC outcome satisfies fair welfare share.
			Under dichotomous preferences, a compelling property is that of \emph{proportional share}~\citep{Dudd15a}. We define it more generally for weak orders as follows. A lottery $p$ satisfies proportional share if for any set $S\subseteq N$, $\sum_{a\in A \midd \exists i\in S \text{s.t. } a\in \max_{\pref_i}(A)}p(a)\geq |S|/n$.
We note that proportional share implies fair share.\footnote{ESR does not satisfy proportional share and the maximal lottery rule does not satisfy fair welfare share.}
			It is easy to establish that RMEC satisfies proportional share. 

	\begin{proposition}
		RMEC satisfies the proportional share property and hence the fair share property. 
		\end{proposition}
	
		A different fairness requirement is that each agent finds the outcome at least as preferred with respect to \sd as the uniform lottery. A PSCF $f$ satisfies \sd-uniformity if for each profile $\pref$, $f(\pref)\pref_i^{\sd} \frac{1}{m}a_1+\cdots+\frac{1}{m}a_m$ for each $i\in N$. 
RMEC does not satisfy \sd-uniformity. 
However, we show that \sd-uniformity is incompatible with very strong \sd-participation. 
	
		\begin{proposition}
			There exists no PSCF that satisfies very strong \sd-participation and \sd-uniformity.
			\end{proposition}
			\begin{proof}
			
				Consider the following preference profile. 
	\begin{align*}
				1:&\quad a, b, c&
				2:&\quad c, b, a&
				3:&\quad a, b, c
	\end{align*}

				When 1 and 2 vote, SD-uniformity demands, that the outcome is uniform.
				When 1, 2, 3 vote, SD-uniformity still demands that the
				outcome is uniform. However very strong-SD-participation demands that
				3 should get strictly better outcome with respect to \sd.
				\end{proof}

	
	Whereas RMEC satisfies the strongest notion of participation, it is vulnerable to strategic misreports.


%
%

\begin{example}
			Consider the following preference profile.
					\begin{align*}
								1:&\quad a, b, c, d, e&
								2:&\quad e,d,c,b,a&
								3:&\quad \{d,c\}, \{a,b,e\}
					\end{align*}

The outcome is $\nicefrac{1}{3}a+\nicefrac{1}{3}d+\nicefrac{1}{3}e$.
Now assume that agent $1$ reports, $\pref_1':\quad a, c, b, e.$
Then the outcome is $\nicefrac{1}{3}a+\nicefrac{1}{3}c+\nicefrac{1}{3}e$ which is better than $\nicefrac{1}{3}a+\nicefrac{1}{3}d+\nicefrac{1}{3}e$  with respect to \sd for agent $1$.
		\end{example}

	On the other hand, if $n \leq 2$, we can prove that RMEC satisfies \sd-strategyproofness. Also if preferences are strict or if they are dichotomous, RMEC is \sd-strategyproof. We note that RMEC is also \sd-strategyproof if preferences are dichotomous or if they are strict. 
		
		We also note that RMEC satisfies a natural monotonicity property: reinforcing an alternative in the agent's preferences can only increase its probability.

			\section{Discussion}
	
			In this paper, we continued the line of research concerning strategic aspects in probabilistic social choice (see e.g.,~\citep{Aziz13b,ABB13d, ABBH12a,BBH15b,Gibb77a,Proc10a}). 
We proposed the RMEC rule that satisfies very strong SD-participation and ex post efficiency as well as various other desirable properties. 
In view of its various properties, it is a useful PSCF with a couple of advantages over RSD. Unlike maximal lotteries~\citep{Bran17a} and ESR~\citep{AzSt14a}, RMEC is relatively simple and does not require linear programming to find the outcome lottery. The use of rank maximality also makes it easier to deal with weak orders in a principled manner.

\paragraph{A general approach}


Consider a scoring vector $s=(s_1,\ldots, s_m)$ such that $s_1>\cdots>s_m$.
An alternative in the $j$-th most preferred equivalence class of an agent is given score $s_j$. An alternative with the highest score is the one that receives the maximum total score from the agents (see for e.g., \citep{FiGe76a} for discussion on positional scoring vectors). Note that an alternative is rank maximal if it achieves the maximum total score for a suitable scoring vector $(n^m,n^{m-1},\ldots, 1)$. We also note that RMEC is defined in a way so that each agent gives $1/n$ probability to his most preferred alternatives that have the best rank vector. The same approach can also be used to select the most preferred alternatives that have the best Borda score or score with respect to any decreasing positional scoring vector. We refer to $s$-MEC as the maximal equal contribution rule with respect to scoring vector $s$. In the rule, each agent identifies $F(i,A,\pref)$ the subset of alternatives in $\max_{\pref_i}(A)$ with the best total score and uniformly distributes $1/n$ among alternatives in $F(i,A,\pref)$.  
The argument for very strong \sd-participation and ex post efficiency still works for any $s$-MEC rule. Any $s$-MEC rule is also anonymous, neutral, single-valued, and proportional share fair.

	\bigskip
	
	It will be interesting to see how RMEC fares on more structured preferences~\citep{AnPo16a}.
	Random assignment rules~\citep{BoMo01a, Kase06a} can be seen as applying a PSCF to a voting problem with more structured preferences~(see e.g., \citep{AzSt14a}). 
It will be interesting to see how RMEC will fare as a random assignment rule especially in terms of \sd-efficiency.


		
  \bibliographystyle{abbrv} 


%

\end{document}